\documentclass{amsart}

\usepackage{mathtools}
\usepackage[square,numbers]{natbib}
\usepackage{float}
\restylefloat{table}
\usepackage[utf8]{inputenc}
\usepackage[toc,page]{appendix}
\usepackage{amsmath,amsthm,amssymb,stmaryrd}
\usepackage{graphicx}
\usepackage[colorinlistoftodos,prependcaption]{todonotes}
\usepackage{xargs}
\usepackage{multicol}
\newcommandx{\unsure}[2][1=]{\todo[linecolor=blue,backgroundcolor=blue!25,bordercolor=blue,#1]{#2}}
\usepackage[section]{placeins}
\usepackage{float}
\usepackage{comment}
\usepackage{xcolor}
\usepackage{graphicx}
\usepackage{dcolumn}
\usepackage{bm}
\usepackage{amsaddr}

\newmuskip\pFqmuskip
\newcommand*\pFq[6][8]{%
	\begingroup 
	\pFqmuskip=#1mu\relax
	\mathcode`\.=\string"8000
	\begingroup\lccode`\~=`\,
	\lowercase{\endgroup\let~}\pFqcomma
	{}_{\,#2}F_{\,#3}{\left[\genfrac..{0pt}{}{\,#4}{\,#5};#6\right]}%
	\endgroup
}
\newcommand{\pFqcomma}{\mskip\pFqmuskip}

\newtheorem{theorem}{Theorem}[section]

\newtheorem{lemma}{Lemma}[section]

\newtheorem{definition}{Definition}[section]

\newcommand{\bbint}[2]{\ensuremath{\;\backslash\!\!\!\!\backslash\!\!\!\!\!\int_{#1}^{#2}}}

\begin{document}
\title[]{Exact Evaluation and resummation of the divergent expansion for the  Heisenberg-Euler Lagrangian}
\author{Christian D. Tica$^{1,2}$}
\author{Philip Jordan D. Blancas$^{3,4}$}
\author{Eric A. Galapon$^3$}
\address{$^1$ SISSA-Scuola Internazionale Superiore di Studi Avanzati, Via Bonomea 265, 34136 Trieste, Italy}
\address{$^2$ ICTP-International Centre for Theoretical Physics, Strada Costiera 11, 34151 Trieste, Italy}
\address{$^3$Theoretical Physics Group, National Institute of Physics, University of the Philippines, Diliman Quezon City, 1101 Philippines}
\address{$^4$Department of Physics, Ateneo de Manila University, Loyola Heights, Quezon City, 1108 Philippines}
\email{eagalapon@up.edu.ph}

\date{\today}

\maketitle
\begin{abstract}
    We devise a novel resummation prescription based on the method of  finite-part integration [Galapon E.A Proc.R.Soc A 473, 20160567(2017)] to perform a constrained extrapolation of the divergent weak-field perturbative expansion for the Heisenberg-Euler Lagrangian to the nonperturbative strong magnetic and electric field regimes. In the latter case, the prescription allowed us to reconstruct the nonperturbative imaginary part from a finite collection of the real expansion coefficients. We also demonstrate the utility of the various equivalent representations of the Hadamard's finite part in deriving the exact closed form for the Heisenberg-Euler Lagrangian from the nonperturbative integral representation. 
\end{abstract}

\section{Introduction}
Employing perturbation theory (PT) for approximating some function or a finite physical observable, $F(\beta)$, around some small real perturbing parameter, $\beta$, often leads to a divergent power series expansion with an alternating sign character,
\begin{equation}\label{mopy}
     F(\beta) \sim  \sum_{n=0}^{\infty}a_n(-\beta)^{n}, \qquad \beta \to 0^+, \qquad a_n > 0,
\end{equation}
which represents the exact solution only in the asymptotic sense \cite{dingle, wong2, olver1997asymptotics}. An optimally truncated asymptotic power series expansion rapidly delivers experimentally relevant approximations to $F(\beta)$ near the perturbative regime $\beta\to 0^+$. Various summability methods and extrapolation procedures \cite{bender1999advanced, jen} such as Borel resummation and Pad\'e approximants are then required to combine the information contained within a finite string of the coefficients, $a_n$, with the known analytic properties of $F(\beta)$ to piece together a reliable reconstruction of the exact solution in nonperturbative regions, $\beta\to 0^-$ or $\beta\to\pm\infty$, along the real line. 
  
Complications arise  when the exact solution possesses elements that are invisible to any finite order of perturbation theory. Borel analysis of the PT series in this case indicate that nonperturbative ambiguities relate to the factorial divergence and a non-alternating sign character of the terms in the expansion \eqref{mopy} for the case $\beta = -\kappa, \kappa > 0$ \cite{le2012large, marucho, fisher}. In particular, nonperturbative contributions arise from the poles of the Borel transform especially those situated along the real axis of the complex Borel plane which render the Borel-Laplace integral ill-defined in the conventional sense \cite{caliceti, calicet_odd}. The systematic treatment and the recovery of nonperturbative ambiguities arising from Borel nonsummability of PT series expansions is a subject in the theory of resurgence \cite{costin2019resurgent, dorigoni, unsal}.

In physics, paradigmatic examples of nonperturbative effects occur in singular eigenvalue problems in single-particle quantum mechanics \cite{le2012large,Jents, stark, instant, anharm, ptsym, yasuta, graffi}. Perturbing potentials such as in the quartic anharmonic oscillator with a negative coupling, the cubic oscillator, and the Stark Hamiltonian, impart an exponentially suppressed imaginary part to the perturbed energy which characterizes the decay of the resulting metastable states via quantum tunneling process. Unlike the real part of the energy, the imaginary part cannot be recovered from a straightforward summation of the nonalternating divergent weak-coupling PT expansion.

The standard approach for carrying out the resummation of divergent nonalternating PT expansions is the Pad\'e-Borel technique \cite{graffi, florio, jent, byorkin, stark2, borell}. 
Pad\'e approximants are used to perform the analytic continuation of the Borel transform of \eqref{mopy} to a neighborhood of the positive real axis.
They are also used to simulate the singularity structure of the Borel transform such as the location of the poles and generate their contributions. Other functions such as the Gauss hypergeometric function, $_2F_1$, and the Meijer G-function are also used in place of the Pad\'e approximants to accommodate more complicated singularity structures such as the presence of branch cuts \cite{mera2018fast, meraprl}. 

None of these prescriptions however allow for the incorporation of the known specific leading-order behavior of the exact solution in the opposite regime, $F\left(\beta{\to\pm\infty}\right)$, along the real line, which is the key for an extrapolation procedure along this direction \cite{weniger1996construction, Tica_royal, wellen, wellen2,suslov2001summing, le1990hydrogen}. In \cite{Tica_royal}, we demonstrated a resummation scheme that allowed us extrapolate alternating divergent perturbative expansions to the nonperturbative region $\beta\to\infty$ along the real line by incorporating the known non-integer algebraic-power leading order behavior of the exact solution, $F(\beta\to\infty) = \beta^{\nu},\,0<\nu<1$. Here, we formulate the resummation prescription to accommodate the case of a logarithmic leading-order non-pertubative behavior as $\beta\to\infty$. A well-known example is the hydrogen atom in a constant magnetic field of magnitude $B$ \cite{le1990hydrogen, hydro, hydro2} where the logarithmic behavior of the binding energy $\mathcal{E}(B)\sim\frac{1}{2} (\ln B)^2$ in the nonperturbative regime $B\to\infty$ becomes essential in large magnetic fields beyond the neutron stars range. This peculiar strong-field behavior renders the resummation of the divergent PT expansion for the energy eigenvalues difficult to perform \cite{cizik}. 

More generally, here we exploit the perturbative data contained in a finite collection of 
coefficients, $a_n>0$, of the expansion \eqref{mopy} to perform a constrained extrapolation to other nonperturbative regions $\beta\to\pm\infty$ and $\beta\to 0^{-}$ along the real line. More importantly, the extrapolation into the region where $\beta=-\kappa, \kappa > 0$, allows us to recover contributions that are undetectable to any finite-order of perturbation theory.   

This paper is organized as follows. In section \ref{sorat} we give a concise discussion of the important results on the various novel formulations of the Hadamard's finite part and the equivalence among its representations resulting from these independent formulations. This equivalence will play a central role in the development of the resummation prescription in the succeeding sections. Section \ref{bigih} will see the first application of these results when we derive the closed forms for both the real and complex Heisenberg-Euler Lagrangian in quantum electrodynamics (QED) by evaluating their respective exact nonperturbative integral representation using the method of finite-part integration devised in \cite{galapon2}.  

We then demonstrate the resummation and extrapolation procedure in section \ref{bigaj} by taking on the divergent weak-field expansion for the Heisenberg-Euler Lagrangian in both purely magnetic and purely electric field background. In the latter, the Heisenberg-Euler lagrangian is complex and the terms in the PT expansion is nonalternating in sign. The prescription in this case allows us to reconstruct the nonperturbative imaginary part which represents the Schwinger effect, a well-known example of nonperturbative effect in QED \cite{schwinger} which characterizes the instability of the quantum vacuum. Finally in section \ref{conclusion}, we summarize our results and provide a possible direction in which to improve the efficacy of the resummation and extrapolation prescription. 

\section{Hadamard's Finite Part}\label{sorat}
In this section, we give a concise discussion on the computation of Hadamard's finite part of divergent integrals with a pole singularity at the origin,
 \begin{equation}\label{miv}
    \int_{0}^{a} \frac{f(x)}{x^{m}} \mathrm{d}x,\qquad m=1,2,\dots, a>0,
\end{equation}
where $f(x)$ is analytic at the origin and $f(0)\neq 0$. Consistent with the expedient canonical definition \cite{monegato2009definitions},  the Hadamard's finite part may be formulated more rigorously as a complex contour integral \cite{galapon2,galapon2016cauchy} or alternatively as a regularized limit at the poles of Mellin transform integrals \cite{regularizedlimit}. The equivalence of these dual representations is central to the results we 
 present here and their applications. 
 
\subsection{Canonical Representation}
The canonical representation of the finite part of the divergent integral \eqref{miv} is obtained by introducing an arbitrarily small cut-off parameter $\epsilon$, $0<\epsilon<a$, to replace the offending non-integrable origin. The resulting convergent integral is grouped into two sets of terms
\begin{equation}\label{definitepart}
\int_{\epsilon}^{a} \frac{f(x)}{x^{m}} \mathrm{d}x = C_{\epsilon}+D_{\epsilon},
\end{equation}
where $C_{\epsilon}$ is the group of terms that possesses a finite limit as $\epsilon\rightarrow 0$, while $D_{\epsilon}$ diverges in the same limit and consists of terms in algebraic powers of $\epsilon$ and $\ln\epsilon$. The finite part of the divergent integral is then defined uniquely by dropping the diverging group of terms $D_{\epsilon}$, leaving only the limit of $C_{\epsilon}$ and assigning the limit as the value of the divergent integral,
\begin{align}\label{finitepart}
\bbint{0}{a} \frac{f(x)}{x^{m}}\mathrm{d} x = \lim_{\epsilon\rightarrow 0} C_{\epsilon} .
\end{align}
The upper limit $a$ can be also be taken to infinity provided  $f(x)x^{-m}$ is integrable at infinity, in which case,
\begin{align}\label{pisik}
    \bbint{0}{\infty}\frac{f(x)}{x^{m}}\mathrm{d}x =     \lim_{a\to\infty}\bbint{0}{a}\frac{f(x)}{x^{m}}\mathrm{d}x.
\end{align}

\subsection{Contour Integral Representation}
A rigorous formulation of the Hadamard's finite part as a complex contour integral may also be derived from the following form equivalent to equation \eqref{finitepart}, 
\begin{align}\label{form2}
\bbint{0}{a} \frac{f(x)}{x^{m}} \mathrm{d}x = \lim_{\epsilon\rightarrow 0} \left[\int_{\epsilon}^{a} \frac{f(x)}{x^{m}} \mathrm{d}x - D_{\epsilon}\right].
\end{align}
This contour integral representation is given in the following lemma.  The full derivation is given as a proof of Theorem 2.2 in \cite{galapon2}.
\begin{lemma}\label{prop1}
Let the complex extension, $f(z)$, of $f(x)$,  be analytic in the interval $[0,a]$. If $f(0)\neq 0$, then 
	\begin{equation}\label{result1}
	\bbint{0}{a}\frac{f(x)}{x^{m}}\mathrm{d}x=\frac{1}{2\pi i}\int_{\mathrm{C}} \frac{f(z)}{z^{m}} \left(\log z-\pi i\right)\mathrm{d}z, \;\; m = 1, 2 \dots 
	\end{equation}
where $\log z$ is the complex logarithm whose branch cut is the positive real axis and $\mathrm{C}$ is the contour straddling the branch cut of $\log z$ starting from $a$ and ending at $a$ itself, as depicted in figure \ref{tear2}. The contour $\mathrm{C}$ does not enclose any pole of $f(z)$.
\end{lemma}

\subsection{Regularized Limit Representation} \label{bilok}
Let $w(z)$ be a function of the complex variable $z$ that is analytic at some domain $D \subseteq \mathbb{C}$ and  $z_0$ be an isolated singularity of $w(z)$ in $D$, then we can define the deleted neighborhood $\delta_{z_0} = D(r,z_0)\setminus z_0$  where $D(r,z_0)$ is an open disk of radius $r$ centered at $z_0$ such that the function $w(z)$ admits the Laurent series expansion 
 \begin{align}\label{gisit}
     w(z) = \sum_{n=-\infty}^{\infty}a_n(z-z_0)^n,
 \end{align}
 where the coefficients $a_n$ are given by 
 \begin{align}
     a_n = \frac{1}{2\pi i} \oint_{ |z-z_0| = r'} \frac{w(z)}{(z-z_0)^{n+1}}\mathrm{d}z
 \end{align}
 for any $r' < r$.  The radius $r$ is bounded by the distance of the nearest singularity of $w(z)$ from $z_0$. The regularized limit of the function $w(z)$ at $z=z_0$ is defined as follows.

\begin{definition}
Let $z_0$ be an interior point in the domain D of $w(z)$. The regularized limit of $w(z)$ as $z\to z_0$, denoted by
\begin{equation}
        \lim^{\times}_{z\to z_0} w(z),
\end{equation}
is the coefficient $a_0$ in the Laurent series expansion \eqref{gisit} of $w(z)$ in a deleted neighborhood of $z_0$.
\end{definition}
In the case when the function $w(z)$ can be rationalized, that is written in the form $w(z) = h(z)/g(z)$ where $h(z)$
and $g(z)$ are both analytic at $z_0$, and $z_0$ is a simple zero of $g(z)$ while $h(z_0)\neq 0$, then the regularized limit is computed as,
\begin{align}\label{gilik}
    \lim_{z\to z_0}^{\times} \frac{h(z)}{g(z)} = \frac{h'(z_0)}{g'(z_0)} - \frac{h(z_0) g''(z_0))}{2(g'(z_0))^2}.
\end{align}
For the case when $g''(z_0) = 0$, this result reduces to a form similar to the L'Hospital's rule, 
\begin{align}\label{hospital}
    \lim_{z\to z_0}^{\times} \frac{h(z)}{g(z)} = \lim_{z\to z_0}\frac{h'(z)}{g'(z)}.
\end{align}
The derivation of this result is given in the proof of Corollary 3.1 in \cite{regularizedlimit}. 

The finite part of the divergent integral \eqref{miv} 
can be extracted from the 
analytic continuation, $\mathcal{M}^*[f(x); s]$, to the whole complex $s$ plane of the Mellin transform, 
\begin{align}\label{mopt}
    \mathcal{M} \left[f(x) ;\,\,s\right] = \int_{0}^{\infty}x^{s-1}f(x)\mathrm{d}x,
\end{align}
provided there is a non-trivial strip of analyticity of the Mellin integral. The case when the upper limit is a finite $a$ follows from equation \eqref{mopt} by considering $f(x) = g(x)\Theta(a-x)$, where $\Theta(x)$ is the Heaviside step function.  Furthermore, if $f(x)$ is analytic at $x=0$, the Mellin transform $\mathcal{M}[f(x); s]$ has at most simple poles along the real line. 

For a positive integer $m$, if $s=1-m$ is a pole of $\mathcal{M}^*[f(x); s]$, then the finite part integral \eqref{miv} is given by the regularized limit of $\mathcal{M}^*[f(x); s]$ at $s=1-m$,
\begin{align}\label{sigaa}
    \bbint{0}{\infty} \frac{f(x)}{x^{m}} \, \mathrm{d}x = \lim^{\times}_{s\to 1-m} \mathcal{M}^*\ \left[f(x) ;\,\,s\right].
\end{align}
In most situations, the analytic continuation can be written in rational form, $\mathcal{M}^*[f(x); s]=h(s)/g(s)$, where $h(1-m)\neq 0$ and $g(1-m)=0$. In this case, it may be possible to choose the rationalization such that $g''(1-m)=0$ so that when $s=1-m$ is a simple pole of $\mathcal{M}^*[f(x); s]=h(s)/g(s)$ the calculation of the regularized limit at $s=1-m$ reduces to equation \eqref{hospital}.

Both the canonical \eqref{finitepart} and regularized limit \eqref{sigaa} representations are used primarily for computing the value of the Hadamard's finite part explicitly. We demonstrate the foregoing discussion by computing the finite part integral \eqref{sigaa} for the case $f(x) = e^{-bx}$ for $b>0$. The finite part integral is the regularized limit at the poles of the analytic continuation of the following Mellin transform integral \cite[p 20]{brychkov2018handbook}, 
\begin{align}
    \mathcal{M} \left[e^{-b x} ;\,\,s\right] &= \int_{0}^{\infty}x^{s-1}e^{-bx}\mathrm{d}x\\\label{cipt}
    &= b^{-s}\Gamma(s), \qquad \mathrm{Re} \,b, \mathrm{Re}\,s > 0.
\end{align}
The analytic continuation of the Mellin transform to the whole complex plane is simply the right-hand side of equation \eqref{cipt},
\begin{align}\label{sif}
    \mathcal{M}^*\left[e^{-b x} ;\,\,s\right] = b^{-s}\Gamma(s).
\end{align}

\begin{figure}
    \centering
    \includegraphics[scale=0.2]{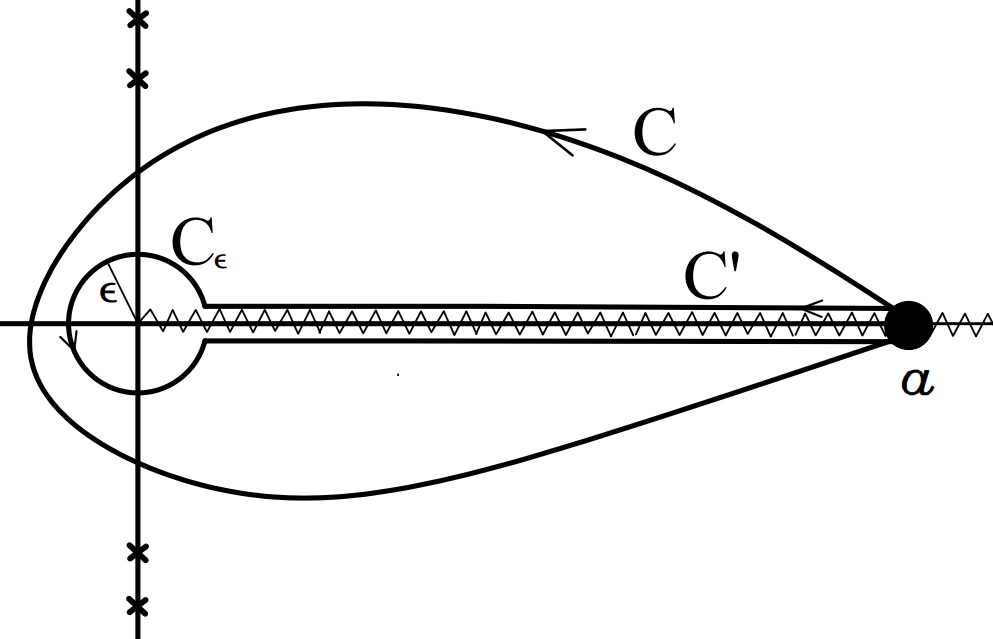}
	\caption{The contour $\mathrm{C}$ used in the representation \eqref{result1} for the Hadamard's finite part. The contour $\mathrm{C}$ excludes any of the poles of $f(z)$. The upper limit $a$ can be infinite. $\epsilon$ is a small positive parameter.}
	\label{tear2}
\end{figure}

Hence, the finite part integral is
\begin{align}\label{miwo}
        \bbint{0}{\infty}\,\frac{e^{-bx}}{x^{m}}\,\mathrm{d}x 
        =\lim^{\times}_{s\to 1-m} b^{-s}\Gamma(s)
        = \lim^{\times}_{s\to 1-m} \frac{\pi b^{-s}}{\sin(\pi s)\Gamma(1-s)}.
\end{align}
where we used the reflection formula $\Gamma(s) = \pi / \sin(\pi s)\Gamma(1-s)$. We then let $g(s) = \sin(\pi s)$ and $h(s) = \pi b^{-s}/ \Gamma(1-s)$ so that from equation \eqref{hospital}, the regularized limit evaluates to
\begin{align}\label{poriy}
    \bbint{0}{\infty}\,\frac{e^{-bx}}{x^{m}}\,\mathrm{d}x = \lim^{\times}_{s\to 1-m} \frac{h(s)}{g(s)} = \lim_{s\to 1-m} \frac{h'(s)}{g'(s)}.
\end{align}
Hence, substituting back $g'(s) = \pi\cos(\pi s)$ and 
\begin{align}
    h'(s) = -\frac{\pi b^{-s}}{\Gamma(1-s)}(\log b-\psi(1-s))
\end{align}
and computing the limit in equation \eqref{poriy}, we obtain for the case when $b$ is real and positive,
\begin{equation}\label{ighi}
    \bbint{0}{\infty}\,\frac{e^{-bx}}{x^{m}}\,\mathrm{d}x = \frac{\left(-1\right)^{m}\,b^{m-1}}{\left(m-1\right)!}\left(\ln\,b-\psi\left(m\right)\right),\qquad m =1,2\dots,
\end{equation}
where $\psi\left(z\right)$ is the digamma function. This is consistent with the result in \cite[eq 3.15]{galapon2} obtained after a lengthy calculation using the canonical definition \eqref{finitepart}.
For the case when $b = -ia$ for any real $a$, a similar calculation yields
\begin{align}\label{hiloy}
    \bbint{0}{\infty}\frac{e^{iax}}{x^{m}}\mathrm{d}x = -\frac{(ia)^{m-1}}{(m-1)!}\left(\ln|a| -\frac{i\pi}{2}\mathrm{sgn}(a)-\psi(m)\right),\qquad m=1,2\dots
\end{align}

In this particular case, the computation of the finite part integral \eqref{ighi} as a regularized limit is more convenient. In general however, especially for cases when the Mellin transform $\mathcal{M}[f(x); s]$ does not exist for some function $f(x)$, the Hadamard's finite part integral can always be obtained from canonical definition \eqref{finitepart} or from its integral representation \eqref{result1}.

\section{ Finite-Part Integration}\label{bigih}
An important application of the contour integral representation \eqref{result1} is a technique known as finite-part integration \cite{galapon2, tica2018finite,tica2019finite,villanueva2021finite} and relies on the equivalence among the three different representations discussed above.  The relevant procedure employed here is a special case of the more general result \cite[eq 64]{galapon3}. We apply it here to give a rigorous derivation of the closed-form for the Heisenberg-Euler Lagrangian from  which exact values can be computed. 

\subsection{The Heisenberg-Euler Lagrangian} The Heisenberg-Euler Lagrangian $\mathcal{L}$ is a nonlinear correction to the Maxwell Lagrangian resulting from the interaction of a vacuum of charged particles of mass $m$ with an external electromagnetic field \cite{schwinger, dunne, dunne_harris,  heisenber_euler, walter, walter2, dittrich, dunne_shw}. In the one-loop order and for the case of constant fields, it depends on the invariant quantity $\beta = e^2 (B^2-E^2)/m^4$, where $e$ is the electron charge. In the case of a purely magnetic background, $E\to0$, it is given for spin-$0$ particles,  
\begin{align}\label{oin}
    \mathcal{L}(\beta) = \frac{m^4}{16\pi^2}f(\beta),
\end{align}
where,
\[
       f(\beta) = \int_{0}^{\infty}\frac{\mathrm{d}\tau}{\tau^3}e^{-\tau}\chi{(\sqrt{\beta}\, \tau)};\qquad\chi(x) = x\,\mathrm{csch}(x) - 1+\frac{x^2}{6}.
\]

This representation is in Heaviside-Lorentz system with natural units so that $h = c =1$ and the fine structure constant reads $\alpha = e^2/4\pi$. For $x\to 0$, the function $\chi(x)$ can be expanded as,
\begin{align}\label{miyp}
    \chi(x) = \sum_{n=2}^{\infty} c_n x^{2n},     \qquad c_n = \frac{2-2^{2n}}{(2n)!}B_{2n},
\end{align}
where $B_{2n}$ are the Bernoulli numbers. Substituting this expansion for $\chi(x)$ into equation \eqref{oin} and integrating term-by-term, we obtain the divergent alternating weak-field expansion for the function $f(\beta)$: 
\begin{equation}\label{gagah}
    f(\beta) = \sum_{n=2}^{\infty} a_n (-\beta)^{n},\qquad a_{n} = (-1)^{n}(2n-3)! c_n, \qquad \beta\to 0.
\end{equation}

For a purely electric background, $B\to 0$, so that $\beta = -\kappa, \kappa = e^2 E^2/m^4$, the Heisenberg-Euler Lagrangian is a complex-valued quantity and admits the following representation,
\begin{align}\label{nitu}
     \mathcal{L}{(\kappa)} = \frac{m^4}{16\pi^2}f(\kappa),\,\,\, f(\kappa) =-\int_{0}^{\infty}\frac{e^{-i \tau}}{\tau^3} \left[\sqrt{\kappa}\,\tau\,\mathrm{csch}\left(\sqrt{\kappa} \tau\right) -1 + \frac{\kappa \tau^2}{6}\right]\mathrm{d}\tau.
\end{align}
The imaginary part signals the important phenomenon of Schwinger effect which predicts the instability of the QED vacuum \cite{dunne_shw}.

The corresponding weak-electric field,  $\kappa\to 0$, perturbation expansion for this representation is identical to \eqref{gagah} but is non-alternating in sign. Since this expansion is real, the imaginary part, which is exponentially vanishing in the weak electric field limit $\kappa\to 0$, is undetectable up to any finite order of this nonalternating weak-field perturbation expansion.

\subsection{Exact Evaluation}

While the integral $f(\beta)$ in the representations \eqref{oin} and \eqref{nitu} are finite and well-defined in both cases, they are expressed as a finite sum of divergent integrals \eqref{miv} with pole singularities at the origin. A formal procedure to arrive at a closed-form is carried out in \cite{ walter, walter2, dittrich} by employing $n$-dimensional regularization scheme to make sense of the divergent integrals. Here, we will make use of the equivalent representations of Hadamard's finite part discussed in section \eqref{sorat} to derive this result more rigorously. 

The initial step is to bring the integration into the complex plane with the use of the contour $\mathrm{C}$ in the representation \eqref{result1} for the Hadamard's finite part,
\begin{align}
    \int_{\mathrm{C}} \frac{g(z)}{z^3}\log z\,\mathrm{d}z,\qquad 
    g(z) = e^{-z} \left(\sqrt{\beta}\,z\, \mathrm{csch} \left(\sqrt{\beta} z\right) - 1 +\frac{\beta z^2}{6}\right),
\end{align}
where the contour $\mathrm{C}$ is given in figure \ref{tear2}. The poles due to $g(z)$ along the imaginary axis are exterior to $\mathrm{C}$ . We deform the contour $\mathrm{C}$ to an equivalent contour $\mathrm{C'}$ so that 
\begin{align}\label{igio}
        \int_{\mathrm{C'}}\frac{g(z)}{z^3}\log z \mathrm{d}z = 2\pi i \int_{\epsilon}^{a} \frac{g(x)}{x^3} \mathrm{d}x + \int_{\mathrm{C}_\epsilon} \frac{g(z)}{z^3} \log z\,\mathrm{d}z,
\end{align}
where $\mathrm{C}_\epsilon$ is the circular contour about the origin. In the limit as $\epsilon\to 0 $, the integral along the circular contour $\mathrm{C_\epsilon}$ vanishes so that equation \eqref{igio} reduces to
\begin{align}\label{nimo}
    \int_{0}^{a} \frac{g(x)}{x^3} \mathrm{d}x = \frac{1}{2\pi i}\int_{\mathrm{C}}\frac{g(z)}{z^3} \log z \mathrm{d}z.
\end{align}

We then add a zero term,
\begin{align}
    -\frac{\pi i}{2\pi i}\int_{\mathrm{C}} \frac{g(z)}{z^3} \mathrm{d}z = -\pi i \sum \mathrm{Res}\,\,\left[\frac{g(z)}{z^3}\right] = 0,
\end{align}
to the right-hand side of equation \eqref{nimo} so that,
\begin{align}
    \int_{0}^{a} \frac{g(x)}{x^3} \mathrm{d}x = \frac{1}{2\pi i}\int_{\mathrm{C}}\frac{g(z)}{z^3}\left(\log z - \pi i \right)\mathrm{d}z
    = \bbint{0}{a}  \frac{g(x)}{x^3}\mathrm{d}x
\end{align}
where we used the contour integral representation \eqref{result1} of the Hadamard's finite part. Hence, taking the limit $a\to\infty$, the integral $f(\beta)$ in the exact representation \eqref{oin} evaluates to
\begin{align}\nonumber
     f(\beta) &= \int_{0}^{\infty}\frac{e^{-\tau}}{\tau^3}\left[\sqrt{\beta}\,\tau \,\mathrm{csch}{\left(\sqrt{\beta}\tau\right)} - 1 + \frac{\beta\tau^2}{6}\right] \mathrm{d}\tau\\\label{gidak}
     &= \sqrt{\beta}\bbint{0}{\infty}\frac{e^{-\tau} \mathrm{csch}{\left(\sqrt{\beta}\tau\right)}}{\tau^2}\mathrm{d}\tau -\bbint{0}{\infty}\frac{e^{-\tau}}{\tau^3}\mathrm{d}\tau + \frac{\beta}{6}\bbint{0}{\infty} \frac{e^{-\tau}}{\tau}\mathrm{d}\tau.
\end{align}

From an expedient point of view, the result \eqref{gidak}, which we obtained rigorously, effectively follows from distributing the integration followed by the immediate regularization of each divergent term as Hadamard's finite part. In section \ref{bigaj}, we will demonstrate that this heuristic will generally result to missing terms especially when one performs a term-by-term integration involving an infinite number of divergent integrals.

The first finite part integral in  right-hand side of equation \eqref{gidak} can be computed using the following Mellin transform integral \cite[p 34, eq 7]{brychkov2018handbook},
\begin{align}\nonumber
    \mathcal{M}\left[e^{-a\tau} \mathrm{csch}{\left(b\,\tau\right)};\,\,s\right] &= \int_{0}^{\infty}\tau^{s-1} e^{-a\tau} \mathrm{csch}{\left(b\,\tau\right)} \mathrm{d}\tau\\\label{igik}
    & = \frac{2^{1-s}}{b^{s}}\Gamma(s)\,\zeta\left(s,\frac{a+b}{2\,b}\right), \qquad \mathrm{Re}\,a > -|\mathrm{Re}\,b|; \mathrm{Re}\,s > 1,
\end{align}
 where $\,\zeta\left(z,\nu\right)$ is the Hurwitz zeta function. The finite part integral is the regularized limit at the simple pole $s=-1$ of the analytic continuation of the Mellin transform to the whole complex plane. The analytic continuation of the Mellin transform is simply given  by the right-hand side of equation \eqref{igik},
 \begin{align}
\mathcal{M}^{\ast}\left[e^{-a\tau} \mathrm{csch}{\left(b\,\tau\right)};\,\,s\right] = \frac{2^{1-s}}{b^{s}}\Gamma(s)\,\zeta\left(s,\frac{a+b}{2\,b}\right).
 \end{align}
 So that
\begin{align}
    \bbint{0}{\infty}\frac{e^{-\tau} \mathrm{csch}{\left(\sqrt{\beta}\tau\right)}}{\tau^2}\mathrm{d}\tau 
    &= \lim_{s\to-1}^{\times} \mathcal{M}^{\ast}\left[e^{-\tau} \mathrm{csch}{\left(\sqrt{\beta}\,\tau\right)};\,\,s\right] \\
    &= \lim_{s\to-1}^{\times} \frac{2^{1-s}}{\beta^{s/2}} \frac{\pi}{\sin(\pi s)\Gamma(1-s)}\,\zeta\left(s,\frac{1+\sqrt{\beta}}{2\sqrt{\beta}}\right).
\end{align}
where we've made use of the reflection formula for $\Gamma(s)$.
We then rationalize the analytic continuation of the Mellin transform by writing,
\begin{align}
\mathcal{M}^{\ast}\left[e^{-\tau} \mathrm{csch}{\left(\sqrt{\beta}\,\tau\right)};\,\,s\right] = \frac{h(s)}{g(s)},
\end{align}
where
\begin{equation}\label{igli}
    h(s) = \frac{2^{1-s}\pi}{\beta^{s/2}\Gamma(1-s)} \zeta\left(s,\frac{1+\sqrt{\beta}}{2\sqrt{\beta}}\right), \qquad g(s) = \sin(\pi s).
\end{equation}
From the formula \eqref{hospital}, we compute the regularized limit as
\begin{align}\label{miytp}
      \bbint{0}{\infty}\frac{e^{-\tau} \mathrm{csch}{\left(\sqrt{\beta}\tau\right)}}{\tau^2}\mathrm{d}\tau = \lim_{s\to-1}^{\times} \frac{h(s)}{g(s)} = \lim_{s\to-1} \frac{h'(s)}{g'(s)}.
\end{align}
Computing the derivatives, $g'(s) = \pi\cos(\pi s)$ and 
\begin{align}\nonumber
    h'(s) = \frac{2^{1-s}\pi\beta^{-s/2}}{\Gamma(1-s)}& \left[\left(\psi(1-s)-\ln\beta-\ln 2\right)\,\zeta\left(s,\frac{1+\sqrt{\beta}}{2\sqrt{\beta}}\right)\right.\\
    + &\left.\,\zeta^{(1,0)}\left(s,\frac{1+\sqrt{\beta}}{2\sqrt{\beta}}\right)  \right],
\end{align}
where $\zeta^{(1,0)}(z,\nu)$ is the derivative of the Hurwitz zeta function with respect to the first argument $z$. Hence, the finite part integral \eqref{miytp} evaluates to 
\begin{align}\nonumber
          \bbint{0}{\infty}\frac{e^{-\tau} \mathrm{csch}{\left(\sqrt{\beta}\tau\right)}}{\tau^2}\mathrm{d}\tau = 2\sqrt{\beta} &\left[\left(\ln\beta+\ln 4+2\gamma-2\right)\,\zeta\left(-1,\frac{1+\sqrt{\beta}}{2\sqrt{\beta}}\right) \right.\\
          &- \left. 2\,\zeta^{(1,0)}\left(-1,\frac{1+\sqrt{\beta}}{2\sqrt{\beta}}\right)\right],
\end{align}
where $\gamma = - \psi(1)$ is the Euler-Mascheroni constant. 

The other finite part integrals in the right-hand side of equation \eqref{gidak} are special cases of equation \eqref{ighi}. Hence, the integral representation \eqref{gidak} for the Heisenberg-Euler Lagrangian in the case of spin-0 particles takes the following closed-form,
\begin{equation}\label{impin}
f(\beta) = \frac{\beta\ln\beta}{12} - \frac{\ln\beta}{4}+\beta\left(\frac{\ln 4}{12}-\frac{1}{6}\right) -\frac{\ln 4}{4}-\frac{1}{4}- 4\beta \zeta^{(1,0)}\left(-1,\frac{1+\sqrt{\beta}}{2\sqrt{\beta}}\right),
\end{equation}
where we made use of the relation \cite{NIST} to simplify the result.
The result \eqref{impin} is consistent with that given in \cite{dunne,walter,walter2,dittrich} using a different approach based on $\zeta$-function regularization. The result given in \cite[eq 3.26]{dittrich} obtained formally using $n$-dimensional regularization has an error and is rectified in \cite{walter}.

Similarly for the case of a purely electric field background, the representation \eqref{nitu} is evaluated as, 
\begin{align}\label{ibiki}
f(\kappa) 
    =-\sqrt{\kappa}\,\bbint{0}{\infty}\frac{e^{-i \tau}\mathrm{csch}(\sqrt{\kappa} \tau)}{\tau^2}\mathrm{d}\tau+\bbint{0}{\infty}\frac{e^{-i \tau}}{\tau^3}\mathrm{d}\tau - \frac{\kappa}{6}\bbint{0}{\infty}\frac{e^{-i \tau}}{\tau} \mathrm{d}\tau.
\end{align}
The first term evaluates to
\begin{align}\nonumber
    \bbint{0}{\infty}\frac{e^{-i \tau}\,\mathrm{csch}(\sqrt{\kappa} \tau)}{\tau^2}\mathrm{d}\tau = -4\sqrt{\kappa}& \left[\left(\psi(2)- \ln \left(2\sqrt{\kappa}\right)\right)\,\zeta\left(-1, \frac{\sqrt{\kappa}+i}{2\sqrt{\kappa}}\right)  \right.\\\label{ciud}
    + &\left.\,\zeta^{(1,0)}\left(-1, \frac{\sqrt{\kappa}+i}{2\sqrt{\kappa}}\right) \right],
\end{align}
while finite part integrals in the right-hand side of equation \eqref{gidak} are special cases of the result \eqref{hiloy}. Hence we obtain the closed form
\begin{align}\nonumber
    f(\kappa) = \frac{\kappa}{6}-\frac{1}{4}-\left(\frac{\kappa}{6}+\frac{1}{2}\right) \ln{\left(2\sqrt{\kappa}\right)} + &i\left(\frac{\pi}{4} + \frac{\kappa\pi}{12}\right) \\\label{wadik}
    &+ 4\kappa\, \zeta^{(1,0)}\left(-1, \frac{\sqrt{\kappa}+i}{2\sqrt{\kappa}}\right).
\end{align}

As suggested in \cite{dunne_harris}, a result numerically consistent with  \eqref{wadik} may be obtained by a suitable analytic continuation, $B\to -iE$, of the corresponding closed-form \eqref{impin} for the case of a magnetic background. This procedure must be carefully implemented by choosing an appropriate branch for the complex logarithmic function and the Hurwitz zeta functions appearing in equation \eqref{impin}. In particular, we find that the branch $\log(iy) = \ln y +i \pi/2 $ is compatible with Mathematica 13.2 's implementation for $\zeta\left(-1, x+iy\right)$ and $\zeta^{(1,0)}\left(-1, x+iy\right)$ for any $x,y>0$, to give a numerical value consistent with that computed from \eqref{wadik}.

\begin{table}
\centering
	\begin{tabular}{ l l l l l l }
		\hline
		$d$   &  $\beta=10^{-2}$ & $\beta=0.1$&  $\beta=0.2$\\ 
		\hline
		1   & $\textcolor{blue}{1.932}143(10^{-6})$ & $\textcolor{blue}{1.8}214(10^{-4})$ & $6.7937(10^{-4})$ \\

		5   & $\textcolor{blue}{1.9323847}854(10^{-6})$ & $\textcolor{blue}{1.83}50(10^{-4})$ &$6.235({10^{-4}})$\\
		9   & $\textcolor{blue}{1.932384796}847(10^{-6})$ & $\textcolor{blue}1.6194(10^{-4})$ 
                &$-4.978(10^2)$ \\
		20  & $\textcolor{blue}{1.9323847969}843(10^{-6})$ & $8.42618(10^{5})$ & $3.636(10^{12})$\\
		50  & $3.3995123(10^4)$ &  &\\
          \hline
		Exact & $ 1.93238479692775525(10^{-6})$ & $ 1.83994677220(10^{-4})$ & $7.0356826048(10^{-4})$ \\
		\hline
  
	\end{tabular}
 
	\caption{Convergence of the partial sums of the perturbative expansion \eqref{gagah} for the integral $f(\beta)$ in equation \eqref{oin} for the Heisenberg-Euler Lagrangian in the case of purely magnetic background. The exact result is computed from the closed-form \eqref{impin} of the integral representation \eqref{oin}.}
 \label{bigak}
\end{table}

\section{Resummation and extrapolation of the weak-field expansion for the Heisenberg-Euler Lagrangian}\label{bigaj}

We now discuss the application of the method of finite-part integration on the resummation of divergent PT expansion and their extrapolation to nonperturbative regions along the real line. In particular, we will piece together a convergent extrapolant that enables us to compute the Heisenberg-Euler Lagrangian in the strong-field limit from the real coefficients of the divergent weak-field expansion.

\subsection{Purely magnetic case}
The divergence of the PT expansion \eqref{gagah} is evident in the leading growth rate of the coefficients, $a_k\sim(2k)!$ as $k\to\infty$ \cite{dunne}. The same is reflected in the results presented in table \ref{bigak} of using partial sums of this expansion to compute $f(\beta)$ for some values of the parameter $\beta$. In addition,
the function $f(\beta)$ possesses the leading-order behavior in the strong-field regime, $\beta\to\infty$ \cite[eq 1.62]{dunne},
\begin{align}\label{mirt}
    f(\beta)\sim\frac{\beta\ln\beta}{12} +\frac{\ln 2}{6}\beta + \dots
\end{align}

On the basis of this information, we sum the divergent PT series \eqref{gagah}
by mapping the first $d+1$ expansion coefficients $a_{n+2}$ to the positive-power moments $\mu_{2n}$ of some positive function $\rho(x)$,
\begin{align}\label{gigh}
    a_{n+2} =\mu_{2n}= \int_{0}^{\infty} x^{2n}\rho(x)\mathrm{d}x, \qquad n = 0,1,\dots, d.
\end{align}
Substituting this to the expansion \eqref{gagah}, the function $f(\beta)$ is summed formally as,
\begin{align}\nonumber
    f(\beta) &= \sum_{n=2}^{\infty} a_n(-\beta)^{n} = \beta^{2}\sum_{n=0}^{\infty} a_{n+2} (-\beta)^{n} = \beta^{2} \sum_{n=0}^{\infty} \int_{0}^{\infty} x^{2n} \rho(x) (-\beta)^{n} \mathrm{d}x\\
    &= \beta^{2}\int_{0}^{\infty}\rho(x) \left(\sum_{n=0}^{\infty} (-\beta x^2)^n\right)\mathrm{d}x = \beta^2 \int_{0}^{\infty} \frac{\rho(x)}{1+\beta x^{2}} \mathrm{d}x = \beta S(\beta).
\end{align}
which is in terms of the Stieltjes integral,
\begin{align}\label{hartoy}
    S(\beta) =  \int_{0}^{\infty} \frac{\rho(x)}{1/\beta + x^{2}} \mathrm{d}x.
\end{align}

The next step is to evaluate $S(\beta)$ in the strong magnetic field regime in a manner that allows one to incorporate the known leading-order strong-field behavior \eqref{mirt}. The relevant expansion for the generalized Stieltjes integral is obtained using the method of finite-part integration. This is given in the following theorem.
\begin{theorem} \label{lemma0}
Let the complex extension, $\rho(z)$, of the real-valued function $\rho(x)$ for real $x$, be entire, then the generalized Stieltjes integral, $S(\beta)$, admits the following exact convergent expansion
        \begin{align}\label{som}
       S(\beta) = \int_{0}^{\infty}\frac{\rho(x)}{1/\beta + x^{2}}\mathrm{d}x = \sum_{n=0}^{\infty} 
        \frac{(-1)^n}{\beta^n} \mu_{-(2n+2)} +\Delta(\beta),
        \end{align}
        where the term $\Delta(\beta)$ is given by
\begin{align}\label{gibad}
    \Delta(\beta) = \frac{\pi \sqrt{\beta}}{4}\left(\rho\left(\frac{i}{\sqrt{\beta}}\right)+\rho\left(\frac{-i}{\sqrt{\beta}}\right)\right)
    +\frac{\sqrt{\beta}\ln\beta}{4 i}\left(\rho\left(\frac{i}{\sqrt{\beta}}\right)-\rho\left(\frac{-i}{\sqrt{\beta}}\right)\right),
\end{align}
and $\mu_{-(2n+2)}$ are the divergent negative-power moments of $\rho(x)$ are interpreted as the Hadamard's finite part integral,
\begin{equation}\label{pigil}
    \mu_{-(2n+2)} = \bbint{0}{\infty}\frac{\rho(x)}{x^{2n + 2}}\mathrm{d}x.
\end{equation}
\end{theorem}

\begin{proof}
       Deform the contour $\mathrm{C}$ to $\mathrm{C'}$ as shown in figure \ref{tear} and perform the following contour integration, 
\begin{equation}\label{tint}
\int_{\mathrm{C'}} \frac{\rho(z)}{1/\beta + z^{2}} \log z\,\mathrm{d}z = (2\pi i) \int_{0}^{a}\frac{\rho(x)}{1/\beta+ x^{2}}\mathrm{d}x + (2\pi i)\sum \mathrm{Res}\left[\frac{\rho(z)\log z}{1/\beta+z^{2}}\right]
\end{equation}
where the integral along the circular loop, $\mathrm{C_\epsilon}$, vanishes as $\epsilon\to 0$. This yields an expression for the original integral along the real line,
\begin{align}\label{mit}
    \int_{0}^{a}\frac{\rho(x)}{1/\beta+ x^{2}}\mathrm{d}x = \frac{1}{2\pi i} \int_{\mathrm{C}} \frac{\rho(z)}{1/\beta + z^{2}} \log z\,\mathrm{d}z  - \sum \mathrm{Res}\left[\frac{\rho(z)\log z}{1/\beta+z^{2}}\right].
\end{align}
We then add a zero to the first term of the right-hand side of equation \eqref{mit}, by adding and subtracting the term 
\begin{align}
    \frac{\pi i}{2\pi i}\int_{\mathrm{C}} \frac{\rho(z)}{1/\beta + z^{2}}\mathrm{d}z = \pi i\sum \mathrm{Res}\left[\frac{\rho(z)}{1/\beta+z^{2}}\right],
\end{align}
so that the first term in the right-hand side of equation \eqref{mit} becomes 
\begin{align}\nonumber\label{musa}
     \frac{1}{2\pi i} \int_{\mathrm{C}} \frac{\rho(z)}{1/\beta + z^{2}} \log z\,\mathrm{d}z  
     = \frac{1}{2\pi i} \int_{\mathrm{C}} \frac{\rho(z)}{1/\beta + z^{2}} (\log z- \pi i)\,\mathrm{d}z  \\
     + \pi i\sum \mathrm{Res}\left[\frac{\rho(z)}{1/\beta+z^{2}}\right].
\end{align}

 \begin{figure}
    \centering
    \includegraphics[scale=0.19]{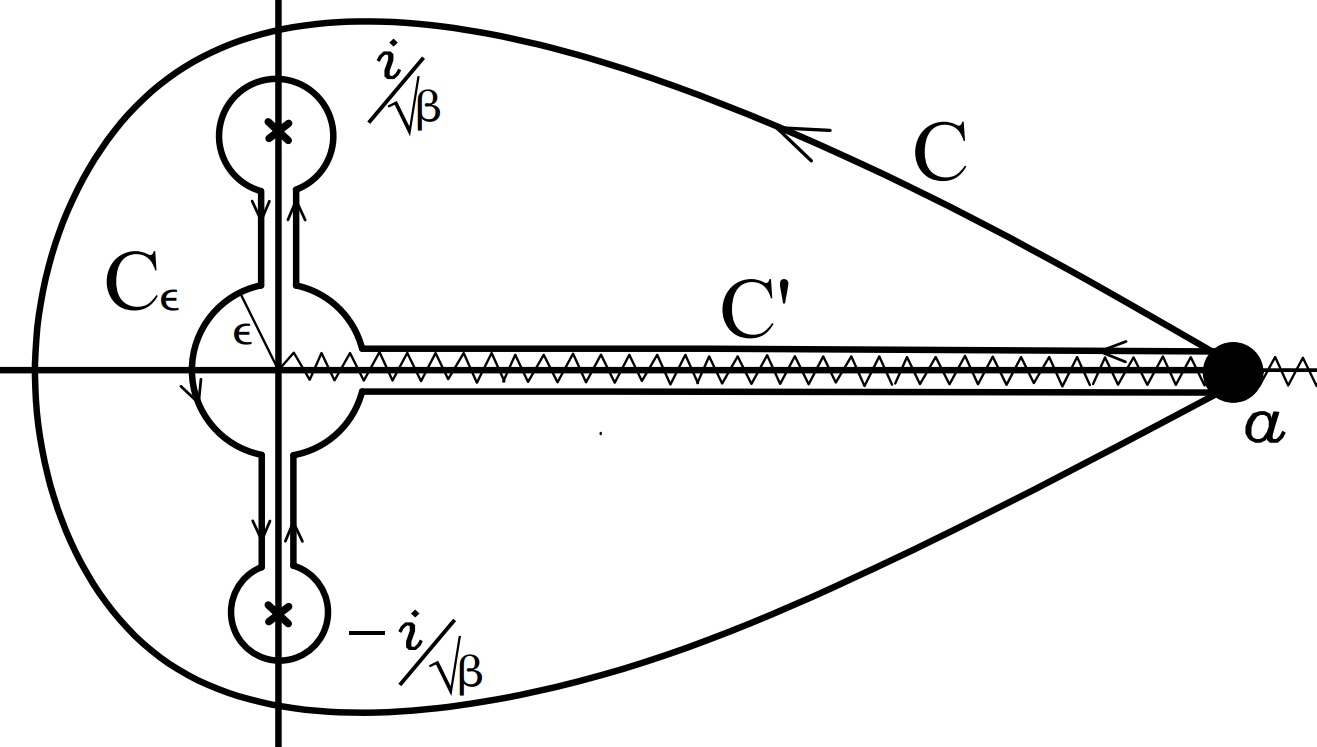}
	\caption{The contour of integration. The upper limit $a$ can be infinite. The poles of the generalized Stieltjes integral at $z=\pm i/\sqrt{\beta}$ lie inside the contour $\mathrm{C}$. $\epsilon$ is a small positive parameter.}
	\label{tear}
\end{figure}
 In the first term of the right-hand of equation \eqref{musa}, we expand $\left(1/\beta + z^{2}\right)^{-1}$ in powers of $1/(\beta z^2)$ and perform a term-by-term integration so that equation \eqref{musa} becomes
\begin{align}\nonumber\label{gabt}
         \frac{1}{2\pi i} \int_{\mathrm{C}} \frac{\rho(z)}{1/\beta + z^{2}} \log z\,\mathrm{d}z =& \sum_{n=0}^{\infty} 
        \frac{(-1)^n}{\beta^n}\frac{1}{2\pi i } \int_{\mathrm{C}}\frac{\rho(z)(\log z-\pi i)}{z^{2n+2}}\mathrm{d}z \\
        &+\pi i\sum \mathrm{Res}\left[\frac{\rho(z)}{1/\beta+z^{2}}\right],\qquad |z| > \frac{1}{\sqrt{\beta}}\\  \label{kinb}
        =&\sum_{n=0}^{\infty} 
        \frac{(-1)^n}{\beta^n} \bbint{0}{a}\frac{\rho(x)}{x^{2n + 2}}\mathrm{d}x + \pi i\sum \mathrm{Res}\left[\frac{\rho(z)}{1/\beta+z^{2}}\right],
\end{align}
where in equation \eqref{kinb}, we used the contour integral representation \eqref{result1} of the Hadamard's finite part integral.   

Substituting equation \eqref{kinb} for the first term of the right-hand side of equation \eqref{mit} and taking the limit as $a\to\infty$,
\begin{align}\label{alad}
\lim _{a\to\infty}\int_{0}^{a}\frac{\rho(x)}{1/\beta + x^{2}}\mathrm{d}x = \sum_{n=0}^{\infty} 
\frac{(-1)^n}{\beta^n} \lim_{a\to\infty}\bbint{0}{a}\frac{\rho(x)}{x^{2n + 2}}\mathrm{d}x +\Delta(\beta).
\end{align}
The term $\Delta(\beta)$ is given by
\begin{align}\nonumber
    \Delta(\beta) = \sum \mathrm{Res}\left[\frac{\rho(z)(\log z - \pi i)}{1/\beta+z^{2}}; z =\pm \frac{i}{\sqrt{\beta}}\right]
\end{align}
which evaluates to equation \eqref{gibad}. Hence, we obtain the result \eqref{som}.

\end{proof}

The expansion in equation \eqref{som} can be shown to be absolutely convergent (the proof is similar to that of Theorem 3.4 in \cite{galapon2}). This convergence and the existence of the limit \eqref{pisik} justifies  interchanging the limit operation with the infinite sum in equation \eqref{alad}. Meanwhile,
the second term $\Delta(\beta)$ is missed by merely performing a formal term-by-term integration followed by a regularization of the divergent integrals by Hadamard's finite part. More importantly, this term provides the dominant behavior for the generalized Stieltjes integral as $\beta\to\infty$. It originates from contributions of the poles $z =\pm i/\sqrt{\beta}$ interior to the contour $\mathrm{C}$ in figure \ref{tear}  which results from the uniformity condition, $|z| > 1/\sqrt{\beta}$, imposed in equation \eqref{gabt}.

By contrast, the finite-part integration \eqref{gidak} of the integral representation for the Heisenberg-Euler Lagrangian involves term-by-term integration over a finite number of divergent integrals so that no uniformity condition is imposed and consequently, the contour $\mathrm{C}$ in figure \ref{tear2} excludes any of the poles of the integrand along the imaginary axis. Hence, we arrive at a real convergent extrapolant for $f(\beta)$ from the result \eqref{som} by finite-part integration,
\begin{equation}\label{hirok}
    f(\beta) = \sum_{k=0}^{\infty} 
\frac{(-1)^k}{\beta^{k-1}} \mu_{-(2k+2)} + \beta \Delta(\beta).
\end{equation}

In order for the expansion \eqref{hirok} to extrapolate the divergent weak-field expansion \eqref{gagah} to the non-perturbative regime, $\beta\to\infty$ along the real line, we incorporate the leading-order behavior \eqref{mirt} through the term $\Delta(\beta)$ given in equation \eqref{gibad}. To this end, we require the reconstruction of the function $\rho(x)$ from the positive-power moments $\mu_{2k}$ in equation \eqref{gigh} to be of the form $\rho(x) = x g(x)$ where $g(0)\neq 0$ and has an entire complex extension $g_s(z)$. This is done by expanding $g(x)$ as a generalized Fourier series expansion in terms of the Laguerre polynomials, $L_m(x)$ \cite{four},
\begin{equation}
    g(x) = \sum_{m=0}^{\infty} c_m\psi_m(x), \,\,\,\, \psi_m(x) = e^{-x/2}L_m(x).
\end{equation}
The basis functions $\psi_m(x)$ obey the orthonormality relation \cite{ortho},
and the Laguerre polynomials are given by \cite{laguerree}, 
\begin{equation}
    L_m(x) = \sum_{k=0}^{m} \frac{\left(-m\right)_k x^{k}}{(k!)^2} = m!\sum_{k=0}^{m}\frac{(-x)^k } {(k!)^2\,(m-k)!} .
\end{equation}
So that the reconstruction of $g(x)$ takes the form,
\begin{equation}\label{mity}
    g(x) = e^{-x/2}\sum_{m=0}^{\infty} c_m m! \sum_{k=0}^{m}\frac{(-x)^k}{(k!)^2 (m-k)!}.
\end{equation}
Hence with the form $\rho(x) = x g(x)$, the second term of the right-hand side of equation \eqref{hirok} can simulate the leading-order behavior \eqref{mirt}. 

The first $d+1$ expansion coefficients, $c_m$, in the reconstruction \eqref{mity} are then computed by imposing the moment condition \eqref{gigh}. This results to a system of linear equations,
\begin{align}\label{sugr}
    a_{n+2} = \sum_{m=0}^{d} c_m P(n,m)
\end{align}
where the matrix $P(n,m)$ is given by
\begin{equation}
    P(n,m) = m! 2^{2n+2}\sum_{k=0}^{m}\frac{(-2)^{k}\,(2n+k+1)!}{(k!)^{2} (m-k)!}.
\end{equation}
We solve this system using the LU factorization method and solver provided by the C++ Eigen 3 library \cite{eigenweb}. We also used arbitrary precision data types from \cite{mpfr} and the C++ Boost Multiprecision libraries to represent the PT coefficients $a_{n+2}$ and perform our computations in arbitrary precision.

\begin{table}
	\begin{tabular}{c lllllll}
		\hline
		Moments & $\beta = 10^{7}$ & $\beta = 10^{12}$ &   $\beta = 10^{13}$ & $\beta = 10^{18}$ \\ 
		\hline

		100  & $\textcolor{blue}{1.07}87(10^{7})$ & $\textcolor{blue}{2.0}424(10^{12})$  & $\textcolor{blue}{2.2}3516(10^{13})$ & $\textcolor{blue}{3.1}991(10^{18})$  \\
  
 	500 & $\textcolor{blue}{1.07}54(10^{7})$ & $\textcolor{blue}{2.03}27(10^{12})$ & $\textcolor{blue}{2.22}416(10^{13})$ & $\textcolor{blue}{3.18}16(10^{18})$  \\ 

		1000 & $\textcolor{blue}{1.076}3(10^{7})$  & $\textcolor{blue}{2.03}47(10^{12})$  & $\textcolor{blue}{2.22}648(10^{13})$  & $\textcolor{blue}{3.18}51(10^{18})$   \\ 

            1500 &  $\textcolor{blue}{1.07}72(10^{7})$  & $\textcolor{blue}{2.036}7(10^{12})$  & $\textcolor{blue}{2.228}60(10^{13})$ & $\textcolor{blue}{3.18}83(10^{18})$\\ 

            2000 & $\textcolor{blue}{1.07}71(10^{7})$ & $\textcolor{blue}{2.036}4(10^{12})$ & $\textcolor{blue}{2.228}33(10^{13})$ & $\textcolor{blue}{3.187}9(10^{18})$  \\
      
            2500 & $\textcolor{blue}{1.0769}4(10^{7})$ & $\textcolor{blue}{2.0361}5(10^{12})$ & $\textcolor{blue}{2.2280}3(10^{13})$ & $\textcolor{blue}{3.1874}5(10^{18})$ \\
        \hline
        $P^{999}_{1000} (\beta)$ & $1.5148(10^{4})$  & $1.5151(10^{9})$ & $1.5151(10^{12})$  & $1.5151(10^{15})$  \\
       
        $P^{49}_{50} (\beta)$ & $1.0723(10^{6})$  & $1.0723(10^{11})$ & $1.0723(10^{12})$  & $1.0723(10^{17})$ \\
        \hline
        $\delta_{499} ( \beta)$ & $8.5224(10^{6})$ & $1.0137(10^{15})$ & $1.0130(10^{17})$ & $1.0129(10^{27})$  \\
        
        $\delta_{100} (\beta)$ & $1.1943(10^{7})$ & $6.1881(10^{16})$ & $6.1880(10^{18})$ & $6.1880(10^{28})$  \\
        \hline
        Exact & $1.07693(10^{7})$ & $2.03613(10^{12})$ & $2.22801(10^{13})$ & $3.18742(10^{18})$ \\
        \hline

	\end{tabular}
 
	\begin{tabular}{c lllllll}

		Moments & $\beta = 1$ & $\beta = 4$ &  $\beta = 10^2$ & $\beta = 10^3$ & $\beta = 10^4$ &  \\ 
		\hline

		100 & $\textcolor{blue}{0.0139}583$ & $\textcolor{blue}{0.149}42$  & $\textcolor{blue}{17}.2803$ & $\textcolor{blue}{32}8.27$ & $\textcolor{blue}{507}1.2$ \\

 	500 & $\textcolor{blue}{0.0139688}5101 $ & $\textcolor{blue}{0.149783}78$ & $\textcolor{blue}{17.35}40$ & $\textcolor{blue}{329}.10$ & $\textcolor{blue}{507}3.5$  \\ 


		1000 & $\textcolor{blue}{0.013968847}5625$ & $\textcolor{blue}{0.149783}54$  & $\textcolor{blue}{17.35}52$ & $\textcolor{blue}{329}.19$ & $\textcolor{blue}{507}5.9$ \\ 

            1500 & $\textcolor{blue}{0.0139688479}565$ & $\textcolor{blue}{0.1497837}32$ & $\textcolor{blue}{17.356}5$  & $\textcolor{blue}{329.2}68$ & $\textcolor{blue}{507}8.1$ \\ 

            2000 & $\textcolor{blue}{0.0139688479}511$ & $\textcolor{blue}{0.14978372}6$ & $\textcolor{blue}{17.356}4$ & $\textcolor{blue}{329.2}60$ & $\textcolor{blue}{5077}.9$\\

        \hline
        $P^{999}_{1000} (\beta)$ & $0.0139688428836$ & $0.149678652$ & $0.63315$ & $13.0395$ & $149.04$ \\
        
         $P^{49}_{50} (\beta)$ & $0.0139668760758$ & $0.147740086$ & $9.88642$ & $106.322$ & $1071.4$ \\
         \hline
        $\delta_{100} ( \beta)$ & $0.0139688479485$ & $0.149783722$ & $17.3563$ & $329.338$ & $4983.8$ \\
        \hline
        Exact & $0.0139688479485$ & $0.149783722$ & $17.3563$ & $329.251$ & $5077.6$ \\
        \hline
	\end{tabular}

	\begin{tabular}{  c l l }
		Moments & $\beta = 0.1 $ & $\beta = 0.2 $    \\ 
		\hline
		50  & \textcolor{blue}{1.83 9}24$(10^{-4})$  &  \textcolor{blue}{7.0}2 337$(10^{-4})$  \\
  
		100 &  \textcolor{blue}{1.83 99}4$(10^{-4})$ & \textcolor{blue}{7.03 5}36$(10^{-4})$  \\
		    
        500 & \textcolor{blue}{1.83 994 677 2}27$(10^{-4})$  & \textcolor{blue}{7.03 568 26}1$(10^{-4})$  \\

        1000 & \textcolor{blue}{1.83 994 677 220 3}47$(10^{-4})$ & \textcolor{blue}{7.03 568 260 4}74$(10^{-4})$    \\ 
        
        1500 & \textcolor{blue}{1.83 994 677 220 367} 084$(10^{-4})$ & \textcolor{blue}{7.03 568 260 484 }22$(10^{-4})$     \\
        
        2000 & \textcolor{blue}{1.83 994 677 220 367 06}4$(10^{-4})$ &  \textcolor{blue}{7.03 568 260 484 1}9$(10^{-4})$    \\
    
        \hline
        $P^{99}_{100} (\beta)$ & 1.83 994 677 220 361 577$(10^{-4})$ & 7.03 568 260 367 885$(10^{-4})$  \\
        \hline
         $\delta_{25} ( \beta)$ & 1.83 994 677 220 367 065$(10^{-4})$  & 7.03 568 260 484 163$(10^{-4})$ \\
        \hline
        Exact & 1.83 994 677 220 367 060$(10^{-4})$  & 7.03 568 260 484 187$(10^{-4})$ \\
        \hline
	\end{tabular}
        
    \caption{Convergence of the extrapolant \eqref{hirok} constructed from the divergent expansion \eqref{gagah} for the integral $f(\beta)$ in the representation \eqref{oin}. $P^{N}_{M} (\beta)$ is a Pade approximant constructed from the divergent expansion \eqref{gagah} at 3000-digit precision using $N+M+1$ of the positive-power moments, $\mu_{2k}$, in equation \eqref{gigh}. We computed $\delta_{n} (\beta)$ from \cite{jen} using $n+1$ moments. The exact result is computed from the closed-form \eqref{impin}. }
	\label{hinglab}
\end{table}

We then substitute the reconstruction $\rho(x) = x g(x)$, where $g(x)$ is given by equation \eqref{mity}, to the expansion \eqref{hirok} so that the first term evaluates to
\begin{align}\label{gipoy}
\sum_{k=0}^{\infty}\frac{(-1)^{k}}{\beta^{k-1}} \mu_{-(2k+2)} = \sum_{k=0}^{\lfloor\frac{d-1}{2}\rfloor} \frac{(-1)^{k}}{\beta^{k-1}} \left(I_k + J_k + L_k \right)
     + \sum_{k=\lfloor\frac{d-1}{2}\rfloor+1}^{\infty} \frac{(-1)^{k}}{\beta^{k-1}} M_k,
\end{align}
where $\lfloor x \rfloor$ is the floor function and the coefficients are given by,
\begin{equation}
    I_k = \sum_{m=0}^{2k}c_m m! \sum_{l=0}^{m}\frac{(-1)^{l}}{(l!)^2 (m-l)!}\bbint{0}{\infty}\frac{e^{-x/2}}{x^{2k+1-l}}\mathrm{d}x,
\end{equation}
\begin{equation}
    J_k = \sum_{m=2k+1}^{d}c_m m! \sum_{l=0}^{2k}\frac{(-1)^{l}}{(l!)^2 (m-l)!}\bbint{0}{\infty}\frac{e^{-x/2}}{x^{2k+1-l}}\mathrm{d}x,
\end{equation}
\begin{equation}
    L_k = \sum_{m=2k+1}^{d} c_m m! \sum_{l=2k+1}^{m}\frac{(-1)^{l}\,(l-2k-1)!\,2^{l-2k}}{(l!)^2 (m-l)!},
\end{equation}
and
\begin{equation}
    M_k = \sum_{m=0}^{d}c_m m! \sum_{l=0}^{m}\frac{(-1)^{l}}{(l!)^2 (m-l)!}\bbint{0}{\infty}\frac{e^{-x/2}}{x^{2k+1-l}}\mathrm{d}x.
\end{equation}
The finite part integrals appearing in theses terms are given by equation \eqref{ighi},
\begin{equation}
    \bbint{0}{\infty}\frac{e^{-x/2}}{x^{2k+1-l}}\mathrm{d}x = \frac{(-1)^{1-l}\left(\frac{1}{2}\right)^{2k-l}}{(2k-l)!}\left(\ln\left(\frac{1}{2}\right)-\psi(2k+1-l)\right).
\end{equation}

The convergence of the expansion \eqref{hirok} across a wide range of magnetic field strengths is summarized in table \ref{hinglab}. The result presented along each row is computed by adding up to $ k = 2d$ terms of the convergent expansion in equation \eqref{gipoy}, where $d+1$ is the number of moments used. As a rule, the working precision in digits at which we carry out the computation equals the number of moments $a_{n+2}$ used as inputs in the system of linear equations \eqref{sugr}. 

In the strong-magnetic field limit, our result reproduces the first few digits of the exact value. In the weak to intermediate field strength, the extrapolant \eqref{hirok} exhibits an excellent agreement with the exact values as well as with the Pad\'e approximant, $P^{N}_{M} (\beta)$, and the nonlinear sequence transformation, $\delta_{n}(\beta) $, from \cite[eq 4]{jen}. $P^{N}_{M} (\beta)$ is constructed from the PT expansion \eqref{gagah} at 3000-digit working precision using $N+M+1$ expansion coefficients while $\delta_{n}(\beta)$ uses $n+1$. Both these latter methods perform well in the weak to intermediate regimes but fail to extrapolate the divergent PT expansion \eqref{gagah} well into the $\beta\to\infty$ regime regardless of how many coefficients are used as inputs in their construction. This can be traced to the inability of these methods to incorporate the precise logarithmic strong-field behaviors \eqref{mirt} of the Heisenberg-Euler Lagrangian. The Pad\'e approximant for instance exhibits an integer power leading-order behavior, $P^{N}_{M}(\beta)\sim\beta^{N-M}$ as $\beta\to\infty$. 

\subsection{Purely electric case}
In the purely electric case, $\beta = -\kappa, \kappa > 0$, so that the perturbation expansion \eqref{gagah} becomes nonalternating,
the moment prescription \eqref{gigh} sums the nonalternating divergent expansion \eqref{gagah} to the formal integral
\begin{align}\label{darn}
    f(\kappa)
    = \kappa^2\,\int_{0}^{\infty} \frac{\rho(x)}{1-\kappa x^{2}} \mathrm{d}x = \kappa H(\kappa).
\end{align}

The integration can be carried out using a suitable deformation of the contour to evade the pole at $z = 1/\sqrt{\kappa}$. Two such prescriptions are shown in  figure \ref{loop} so that as $\epsilon\to 0$, we obtain
\begin{align}\label{ugay}
         H(\kappa)  = \mathrm{PV}\int_{0}^{\infty}\frac{\,\rho(x)}{1/\kappa- x^{2}}\mathrm{d}x \pm i\frac{\pi  \sqrt{\kappa}}{2} \rho\left(\frac{1}{\sqrt{\kappa}}\right),
\end{align}
where the Cauchy principal value integral is given by
\begin{align}\label{pakot}
\mathrm{PV}\int_{0}^{\infty}\frac{\rho(x)}{1/\kappa - x^2}\mathrm{d}x = \lim_{\epsilon\to 0} \left[ \int_{0}^{\frac{1}{\kappa} -\epsilon}\,\frac{\rho(x)}{1/\kappa - x^2}\mathrm{d}x + \int_{\frac{1}{\kappa} + \epsilon}^{\infty}\,\frac{\rho(x)}{1/\kappa - x^2}\mathrm{d}x \right].
\end{align}
The $\pm$ sign in equation \eqref{ugay} corresponds to performing the integration along the path $\mathrm{C}^{\pm}$. Yet another possibility is to deform the contour in a manner that renders the imaginary part in \eqref{ugay} extraneous so that the value of $H(\kappa)$ is real and is given by the principal value integral \eqref{pakot}.

In general, the manner in which this is done is unclear from perturbation theory alone and a separate prescription must be provided to clear this nonperturbative ambiguity. 
In the case of the Heisenberg-Euler Lagrangian, the imaginary part gives the particle-antiparticle pair production rate, $\Gamma_{\text{prod}}(\kappa) = 2\mathrm{Im}\left(f_{s}(\kappa)\right)$, so that the appropriate contour is $\mathrm{C^+}$. Finally, we perform finite-part integration of the principal value integral in equation \eqref{ugay} as shown in the following theorem.
\begin{figure}
    \centering
    \includegraphics[scale=0.25]{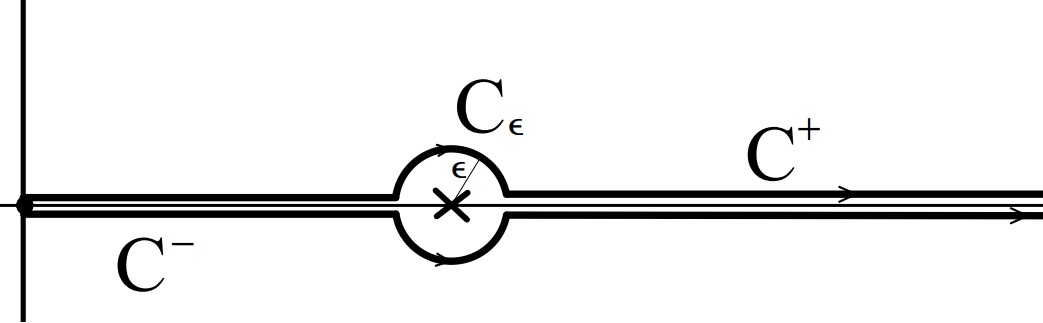}
	\caption{The contours of integration, $\mathrm{C}^{+}$ and $\mathrm{C}^{-}$ used in the prescription for evaluating the integral in \eqref{darn}. They correspond to two directions with which the half-circle $\mathrm{C}_{\epsilon}$  about the pole $z = 1/\sqrt{\kappa}$ is traversed. }
	\label{loop}
\end{figure}
\begin{theorem} \label{lemma3}
Let the complex extension, $\rho(z)$, of the real-valued function $\rho(x)$ along the real line, be entire, then
\begin{align}\label{igar}
       \mathrm{PV}\int_{0}^{\infty}\frac{\rho(x)}{1/\kappa - x^{2}}\mathrm{d}x = -\sum_{n=0}^{\infty} 
        \frac{\mu_{-(2n+2)}}{\kappa^n}  -\Lambda(\kappa),
        \end{align}
        where the term $\Lambda(\kappa)$ is given by
\begin{align}\label{igty}
    \Lambda(\kappa) =  \frac{\sqrt{\kappa}}{2}\ln\left(\sqrt{\kappa}\right)\left(\rho\left(\frac{1}{\sqrt{\kappa}}\right)-\rho\left(-\frac{1}{\sqrt{\kappa}}\right)\right).
\end{align}
and $\mu_{-(2k+2)}$, the divergent negative-power moments of $\rho(x)$, are Hadamard's finite part integrals,
\begin{equation}
    \mu_{-(2n+2)} = \bbint{0}{\infty}\frac{\rho(x)}{x^{2n + 2}}\mathrm{d}x.
\end{equation}
\end{theorem}
\begin{proof}
    Deform the contour $\mathrm{C}$ to $\mathrm{C'}$ as shown in figure \ref{contour3} and perform the following contour integration, 
\begin{align}\nonumber
    \int_{\mathrm{C'}}\frac{\rho(z)}{1/\kappa-z^2}\,\log z\mathrm{d}z = 2\pi i\,\mathrm{PV}&\int_{0}^{a}\frac{\rho(x)}{1/\kappa-x^2}\mathrm{d}x +\int_{\mathrm{C}_{\epsilon_1}}\frac{\rho(z)}{1/\kappa -z^2}\log z\mathrm{d}z \\\label{silop}
    & + \int_{\mathrm{C}_{\epsilon_2}}\frac{\rho(z)}{1/\kappa -z^2}\log z\mathrm{d}z,
\end{align}
where the contour integral around the circular loop $\mathrm{C_\epsilon}$ vanishes as $\epsilon\to 0$ and  the first term of the right-hand side of equation \eqref{silop} containing the desired principal value integral is the sum of the integrals along the segments labelled $1,2,5$ and $6$ in the limit as $\epsilon, \epsilon_2\to 0$. The remaining integrals in the right-hand side of equation \eqref{silop} are evaluated as
\begin{align}\nonumber
\lim_{\epsilon_2 \to 0}\int_{\mathrm{C}_{\epsilon_2}}\frac{\rho(z)}{1/\kappa -z^2}\log z\mathrm{d}z
    = - (2\pi i) &\left[\frac{\sqrt{\kappa}}{2}\rho\left(\frac{1}{\sqrt{\kappa}}\right)\ln\left(\frac{1}{\sqrt{\kappa}}\right)\right. \\\label{gihar}
     &+\left. \pi i \mathrm{Res}\left[\frac{\rho(z)}{z^2-1/\kappa};z=\frac{1}{\sqrt{\kappa}}\right]\right].
\end{align}
and 
\begin{align}\nonumber
   \lim_{\epsilon_1 \to 0} \int_{\mathrm{C}_{\epsilon_1}}\frac{\rho(z)}{1/\kappa -z^2}\log z\mathrm{d}z 
    = (2\pi i) &\left[\frac{\sqrt{\kappa}}{2}\rho\left(-\frac{1}{\sqrt{\kappa}}\right)\ln\left(\frac{1}{\sqrt{\kappa}}\right)\right. \\\label{gihat}
    &- \left.\pi i \mathrm{Res}\left[\frac{\rho(z)}{z^2-1/\kappa};z=-\frac{1}{\sqrt{\kappa}}\right]\right].
\end{align}

\begin{figure}
\centering\includegraphics[scale=0.25]{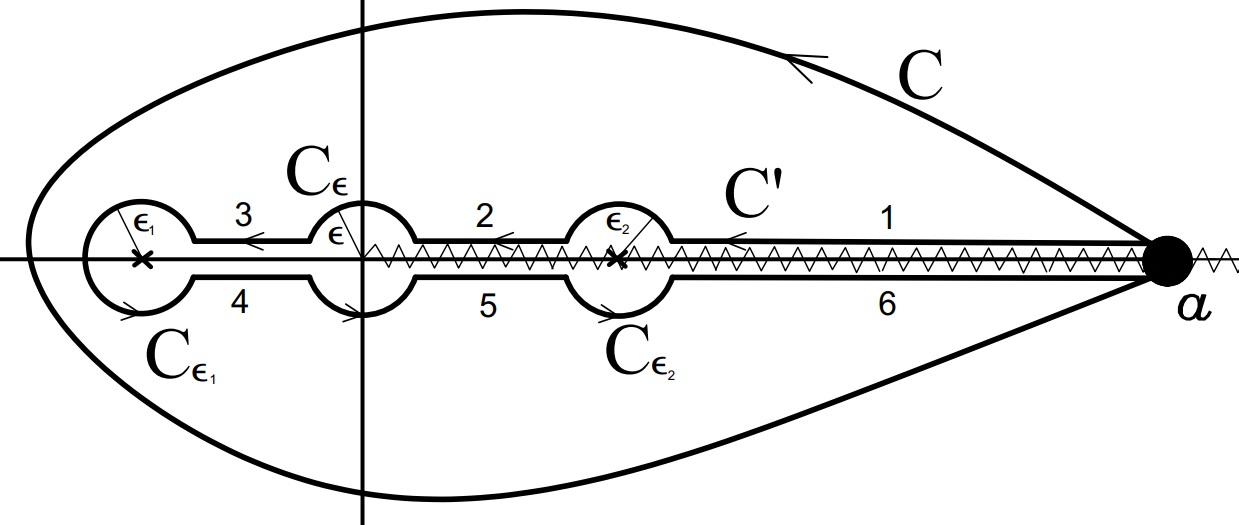}
	\caption{The contour of integration used for the evaluation of the principal value integral \eqref{igar}. Poles located at $z=\pm 1/\sqrt{\kappa}$ are in the interior of the contour $\mathrm{C}$.  }
	\label{contour3}
\end{figure}

Substituting these to equation \eqref{silop} and solving for the principal value integral we obtain,
\begin{align}\nonumber
    \mathrm{PV}\int_{0}^{a}\frac{\rho(x)}{1/\kappa -x^2}\mathrm{d}x = \frac{1}{2\pi i} &\int_{\mathrm{C}}\frac{\rho(z)(\log z -\pi i)}{1/\kappa-z^2}\mathrm{d}z\\\label{smoke}
    &+\frac{\sqrt{\kappa}}{2}\ln\left(\frac{1}{\sqrt{\kappa}}\right)\left(\rho\left(\frac{1}{\sqrt{\kappa}}\right)-\rho\left(-\frac{1}{\sqrt{\kappa}}\right)\right),
\end{align}
where we made use of  
\begin{align}\label{hika}
    \pi i \sum\mathrm{Res}\left[\frac{\rho(z)}{z^2-1/\kappa};z=\pm\frac{1}{\sqrt{\kappa}}\right]= - \frac{\pi i}{2\pi i} \int_{\mathrm{C}}\frac{\rho(z)}{1/\kappa-z^2}\mathrm{d}z.
\end{align}
In the first term in the right-hand side of equation \eqref{smoke}, we expand $(1/\kappa-z^2)^{-1}$ in powers of $1/(\kappa z^2)$ and implement a term-by-term integration,
\begin{align}
     \frac{1}{2\pi i} \int_{\mathrm{C}}\frac{\rho(z)(\log z -\pi i)}{1/\kappa-z^2}\mathrm{d}z =& -\sum_{n=0}^{\infty}\frac{1}{\kappa^n} \frac{1}{2\pi i}\int_{\mathrm{C}}\frac{\rho(z)(\log z -\pi i)}{z^{2n+2}}\mathrm{d}z\\\label{digat}
     = &-\sum_{n=0}^{\infty} 
        \frac{1}{\kappa^n} \bbint{0}{a}\frac{\rho(x)}{x^{2n + 2}}\mathrm{d}x,\,\, |z| > \frac{1}{\sqrt{\kappa}}
\end{align}
where in equation \eqref{digat}, we used the contour integral representation \eqref{result1} of the Hadamard's finite part integral. Substituting equation \eqref{digat} for the first term of the right-hand side of equation \eqref{smoke} and taking the limit as $a\to\infty$,
\begin{align}\nonumber
    \lim_{a\to\infty}\mathrm{PV}\int_{0}^{a}\frac{\rho(x)}{1/\kappa-x^2}\mathrm{d}x = &-\sum_{n=0}^{\infty} 
        \frac{1}{\kappa^n}\lim_{a\to\infty}\bbint{0}{a}\frac{\rho(x)}{x^{2n + 2}}\mathrm{d}x\\\label{lopy}
    &+\frac{\sqrt{\kappa}}{2}\ln\left(\frac{1}{\sqrt{\kappa}}\right)\left(\rho\left(\frac{1}{\sqrt{\kappa}}\right)-\rho\left(-\frac{1}{\sqrt{\kappa}}\right)\right),
\end{align}
This proves the result \eqref{igar}.
\end{proof}

Equivalently, the corresponding extrapolant for the complex Heisenberg-Euler Lagrangian can be obtained by performing an analytic continuation of the real expansion \eqref{hirok} in the purely magnetic case to the negative real line,  $f\left(\beta\to -\kappa\right)$, which is a branch cut of the Stieltjes integral \eqref{hartoy}. The second term, $\Delta(\beta)$, in the expansion \eqref{hirok} then becomes
\begin{align}
\Delta(\beta \to -\kappa) = \pm\frac{i \pi \sqrt{\kappa}}{4}&\left[\rho\left(\frac{1}{\sqrt{\kappa}}\right)+\rho\left(\frac{-1}{\sqrt{\kappa}}\right)\right]
    \\\nonumber
&+\frac{\sqrt{\kappa}\left(\ln\kappa\pm i\pi \right)}{4 }\left[\rho\left(\frac{1}{\sqrt{\kappa}}\right)-\rho\left(\frac{-1}{\sqrt{\kappa}}\right)\right].
\end{align}

The sign ambiguity in this case corresponds to whether the negative real axis is approached from above $ \beta \to e^{i\pi}\kappa$ or below $\beta \to e^{-i\pi}\kappa$. Choosing the latter, we obtain the following complex extrapolant for the divergent expansion \eqref{gagah},
\begin{equation}\label{migyo}
    f(\kappa) = -\sum_{n=0}^{\infty} 
\frac{\mu_{-(2n+2)}}{\kappa^{n-1}} - \kappa \Lambda(\kappa) + \frac{i \pi}{2}\kappa^{3/2}\rho\left(\frac{1}{\sqrt{\kappa}}\right),
\end{equation}
where the term $\Lambda(\kappa)$ encodes the dominant behavior in the strong electric field limit $\kappa\to\infty$
\begin{align}\label{jotro}
    \Lambda(\kappa) =  \frac{\sqrt{\kappa}}{2}\ln\left(\sqrt{\kappa}\right)\left(\rho\left(\frac{1}{\sqrt{\kappa}}\right)-\rho\left(-\frac{1}{\sqrt{\kappa}}\right)\right).
\end{align}

\begin{table}
\centering
  	\begin{tabular}{  c l l l l l   }
   \hline
		Moments & $ \kappa = 1 $ & $\kappa = 4$  \\ 
		\hline
        200 & $\textcolor{blue}{0.0209}01957+i\textcolor{blue}{0.0136}6231$ &  $\textcolor{blue}{0.102}1807 +i\textcolor{blue}{0.25}054011$\\
        500 & $\textcolor{blue}{0.02094}3108+i\textcolor{blue}{0.0136}10668$ & $\textcolor{blue}{0.1022}391 +i\textcolor{blue}{0.2520}3733$\\
        1000 & $\textcolor{blue}{0.02094}3004 +i\textcolor{blue}{0.013609}780$ & $\textcolor{blue}{0.1022}316+i\textcolor{blue}{0.2520}5760$\\
    	1500 & $\textcolor{blue}{0.02094296}2 +i\textcolor{blue}{0.013609}607$  & $\textcolor{blue}{0.10222}66 +i\textcolor{blue}{0.25206}534$ \\
        2000 & $\textcolor{blue}{0.02094296}8 + i\textcolor{blue}{0.013609}603$ & $\textcolor{blue}{0.10222}63 +i\textcolor{blue}{0.252064}87$\\
        \hline
        $P^{999}_{1000} (\kappa)$ & $0.040044268 $  & $-0.296120$  \\
        $P^{49}_{50} (\kappa)$ & $0.034964576$ & $0.2379596$\\
        \hline
        $\delta_{499} (\kappa)$ & $0.012249109$ & $0.2716414$ \\
        
        $\delta_{50} (\kappa)$ & $0.029247263$ & $1.6506276$ \\
        \hline
        
        Exact & $0.020942969 + i 0.013609598$ & $0.1022258 + i0.25206464$ \\
        \hline
	\end{tabular}

	\begin{tabular}{  c l c l l l   }
		Moments & $ \kappa = 0.2$   \\ 
		\hline
        200 & $\textcolor{blue}{9.19}16603636(10^{-4})+i\textcolor{blue}{5.6}19833452(10^{-5})$\\
        500 & $\textcolor{blue}{9.1952}503862(10^{-4}) + i\textcolor{blue}{5.66}2175099(10^{-5})$ \\

        1000 & $\textcolor{blue}{9.19523}67536(10^{-4}) +i\textcolor{blue}{5.6616}82562(10^{-5})$\\
    	1500  &  $\textcolor{blue}{9.1952336}756(10^{-4}) + i\textcolor{blue}{5.661679}405(10^{-5})$     \\
        2000 & $\textcolor{blue}{9.1952336}115(10^{-4}) + i\textcolor{blue}{5.661679}403(10^{-5})$ \\
        \hline
        $P^{999}_{1000} (\kappa)$ & $9.6404911601(10^{-4})$  &  \\
        $P^{49}_{50} (\kappa)$ & $1.0211638800(10^{-3})$\\
        \hline
        $\delta_{499} (\kappa)$ & $1.0414341101(10^{-3})$ & \\
        $\delta_{50} (\kappa)$ & $1.0248313456(10^{-3})$\\
        \hline
        Exact &$9.1952336091(10^{-4}) + i 5.661679704(10^{-5})$  \\
        \hline
	\end{tabular}

    \caption{Convergence of the complex extrapolant \eqref{migyo} constructed from the nonalternating divergent expansion \eqref{gagah} for $f(\kappa)$. The exact result is computed from the closed-form \eqref{wadik}. }
    \label{smolzerospin}
\end{table}

\begin{table}
\centering
   	\begin{tabular}{  c l l l l l   }
        \hline
	Moments & $ \kappa = 10^{8} $ & $\kappa = 10^{12}$  \\ 
		\hline
        200 & $-\textcolor{blue}{1}.30356(10^{8})+i\textcolor{blue}{2}.70526(10^{7}) $ & $-\textcolor{blue}{2.0}969(10^{12})+i\textcolor{blue}{2}.70604(10^{11}) $  \\

        500 & $-\textcolor{blue}{1.26}671(10^{8})+i\textcolor{blue}{2.61}196(10^{7})$ & $-\textcolor{blue}{2.03}27(10^{12})+i\textcolor{blue}{2.61}264(10^{11})$  \\

        1000 & $-\textcolor{blue}{1.26}784(10^{8})+i\textcolor{blue}{2.61}519(10^{7})$ & $-\textcolor{blue}{2.03}48(10^{12})+i\textcolor{blue}{2.61}588(10^{11})$ \\
        
    	1500 & $-\textcolor{blue}{1.268}88(10^{8})+i\textcolor{blue}{2.61}814(10^{7})$  & $-\textcolor{blue}{2.036}7(10^{12})+i\textcolor{blue}{2.61}883(10^{11})$  \\
     
        2000 & $-\textcolor{blue}{1.268}75(10^{8})+i\textcolor{blue}{2.617}76(10^{7})$ & $-\textcolor{blue}{2.036}4(10^{12})+i\textcolor{blue}{2.61}845(10^{11})$  \\
        \hline
        
        $P^{999}_{1000} (\kappa)$ & $-1.5151(10^{6})$  & $-1.5151(10^{11})$  \\
        $P^{49}_{50} (\kappa)$ & $-1.0723(10^{7})$ & $-1.0723(10^{11})$\\
        \hline
        $\delta_{499} (\kappa)$ & $-7.4396(10^{7})$  & $1.012(10^{15})$\\
        $\delta_{50} (\kappa)$ & $-5.3991(10^{8})$ & $-4.93990(10^{16})$ \\
        \hline
        Exact & $-1.26859(10^{8})+i2.61730(10^{7})$ & $-2.0361(10^{12})+i2.61799(10^{11})$ \\
	\end{tabular}

   	\begin{tabular}{  c l l l l l   }
        \hline

		Moments & $ \kappa = 10 $ & $\kappa = 100$ &  $\kappa=10^{3}$ \\ 
		\hline
        200 & $0.046028+i\textcolor{blue}{1.08}0516 $ & $-\textcolor{blue}{12}.8842 +i\textcolor{blue}{}20.29965$ & $-\textcolor{blue}{31}4.188+i\textcolor{blue}{24}7.16$ \\
        500 & $\textcolor{blue}{0.038}571+i\textcolor{blue}{1.084}611$ & $-\textcolor{blue}{1}3.0081+i\textcolor{blue}{19.9}8976 $ & $-\textcolor{blue}{311.02}1 +i\textcolor{blue}{240}.37$\\
        1000 & $\textcolor{blue}{0.0384}62+i\textcolor{blue}{1.084}547$ & $-\textcolor{blue}{1}3.0022 + i\textcolor{blue}{19.99}103$ & $-\textcolor{blue}{311.0}18+i\textcolor{blue}{240}.53$\\
        1500 & $\textcolor{blue}{0.0384}08 +i\textcolor{blue}{1.08449}4 $ & $-\textcolor{blue}{12.99}73+i\textcolor{blue}{19.992}68$ &   $-\textcolor{blue}{311.02}1+i\textcolor{blue}{240.6}9$   \\
        2000 & $\textcolor{blue}{0.03841}4 +i\textcolor{blue}{1.08449}3$ & $-\textcolor{blue}{12.99}77+i\textcolor{blue}{19.992}72$ & $-\textcolor{blue}{311.02}4 +i\textcolor{blue}{240.6}7 $\\
        \hline
        $P^{999}_{1000} (\kappa)$ & $0.948117 $  & $-18.2559$ & $-154.067$ & \\
        $P^{49}_{50} (\kappa)$ & $-13.0793$ & $-11.7271$ & $-108.15$\\
        \hline
        $\delta_{499} (\kappa)$ & $-0.59973$ & $-11.7680 $ & $7542.32$\\
        $\delta_{50} (\kappa)$ & $-4.67879$& $207.549$ & $-9595.11$\\
        \hline
        Exact & $0.038419 + i1.084491$ & $-12.9983  + i19.99284$ & $-311.028 +i 240.65$ \\
        \hline
	\end{tabular}

    \caption{Convergence of the extrapolant \eqref{migyo} for $f(\kappa)$ in the  Heisenberg-Euler Lagrangian in the strong electric field regimes. The exact result is computed from the closed-form \eqref{wadik}.}
    \label{bigzerospin}
\end{table}

The convergence of the complex extrapolant \eqref{migyo} across a wide range of electric field strengths is summarized in tables \ref{smolzerospin} and \ref{bigzerospin}. We also included the results of the Pad\'e approximants $P^N_{M}(\kappa)$ and the non-linear sequence transformation $\delta_n(\kappa)$ from \cite[eq 4]{jen}. Regardless of the orders at which these approximants are constructed, both become unreliable at summing and extrapolating the non-alternating weak-field expansions for the Heisenberg-Euler Lagrangian beyond perturbations of order $\kappa = 10^{-1}$.  More importantly, on their own, they do not offer any means by which the non-perturbative imaginary parts of the Heisenberg-Euler Lagrangians can be recovered from a finite collection of the real coefficients of the corresponding divergent expansion.

\section{Conclusion}\label{conclusion}\label{summ}
In this paper, we proposed a prescription based on the method of finite-part integration for the resummation and extrapolation of divergent PT series expansions with coefficients that we map to the positive-power moments, $\mu_{2n} =\int_{0}^{\infty}x^{2n}\rho(x)\mathrm{d}x$, of some positive function $\rho(x)$. We applied the procedure on the divergent  weak-field expansions of the exact integral representations for the Heisenberg-Euler Lagrangian for both purely magnetic and purely electric background. In each of these examples, the procedure allowed us to transform the weak-field expansion into a novel convergent expansion in inverse powers of perturbation parameter plus a correction term that led us to incorporate the known logarithmic leading-order behavior in the strong-field regime. The prescription also recovered the non-perturbative imaginary parts from the real coefficients of the divergent expansions. This enabled us to construct extrapolants which can be used across a wide range of values for the field strength with considerable accuracy. Furthermore, we also showed how the method of finite-part integration can be used to evaluate in closed-form the exact integral representations for both the real and complex Heisenberg-Euler Lagrangians. 

An important feature to note about the dominant term $\Delta(\beta)$ in equation \eqref{gibad} as well as 
the second and third terms in the extrapolants \eqref{migyo} which simulate the leading behaviour \eqref{mirt} is that they contain the sampling of the reconstruction $\rho(x)$ near the origin as $\beta,\kappa\to \infty$. In this regard, the resummation prescription might be improved by using a more parsimonious prescription for solving the underlying Stieltjes moment problem with a better point-wise convergence near the origin. A promising alternative proposed in \cite{mead1984maximum} is an information-theoretic approach and is based on the maximization of the Shannon entropy functional.  This could enhance the applicability of the resummation procedure on problems with only a few accessible perturbative coefficients. 

\section*{Acknowledgments}
We acknowledge the Computing and Archiving Research Environment (COARE) of the Department of Science and Technology's Advanced Science and Technology Institute (DOST-ASTI) for providing access to their High-Performance Computing (HPC) facility. This work is funded by the University of the Philippines System through the Enhanced Creative Work Research Grant (ECWRG 2019-05-R). C.D. Tica acknowledges the Department of Science and Technology-Science Education Institute (DOST-SEI) for the scholarship grant under DOST ASTHRDP-NSC.

\end{document}